\documentclass[sigconf]{acmart}

\AtBeginDocument{%
  \providecommand\BibTeX{{%
    \normalfont B\kern-0.5em{\scshape i\kern-0.25em b}\kern-0.8em\TeX}}}

\copyrightyear{2022}
\acmYear{2022}
\setcopyright{acmlicensed}\acmConference[WWW '22]{Proceedings of the ACM Web Conference 2022}{April 25--29, 2022}{Virtual Event, Lyon, France}
\acmBooktitle{Proceedings of the ACM Web Conference 2022 (WWW '22), April 25--29, 2022, Virtual Event, Lyon, France}
\acmPrice{15.00}
\acmDOI{10.1145/3485447.3511949}
\acmISBN{978-1-4503-9096-5/22/04}

\usepackage{booktabs}
\usepackage{multirow}
\usepackage{tabularx}
\usepackage{bm}

\newtheorem{thm}{Theorem}

\begin{document}

\title{Modeling User Behavior with Graph Convolution for Personalized Product Search}




\author{Lu Fan$^{+}$, Qimai Li$^{+}$, Bo Liu, Xiao-Ming Wu} 
\authornote{Corresponding author. \\+ Equal contribution. }
\affiliation{%
  \institution{The Hong Kong Polytechnic University}
  \city{Hong Kong}
  \country{China}
  \institution{\\ \{cslfan, csqmli, csbliu, csxmwu\}@comp.polyu.edu.hk}
}
\author{Xiaotong Zhang}
\affiliation{%
  \institution{Dalian University of Technology}
  \city{Dalian}
  \country{China}
}
\email{zxt.dut@hotmail.com}

\author{Fuyu Lv, Guli Lin, Sen Li}
\author{Taiwei Jin, Keping Yang}
\affiliation{%
  \institution{Alibaba Group}
  \city{Hangzhou}
  \country{China}
}
\email{lvfuyu91@sina.com, linguli@gmail.com}
\email{lisen.lisen@alibaba-inc.com}
\email{taiwei.jtw@alibaba-inc.com}
\email{shaoyao@taobao.com}


\renewcommand{\shortauthors}{Fan and Li, et al.}

\begin{abstract}

User preference modeling is a vital yet challenging problem in personalized product search. In recent years, latent space based methods have achieved state-of-the-art performance by jointly learning semantic representations of products, users, and text tokens. However, existing methods are limited in their ability to model user preferences. They typically represent users by the products they visited in a short span of time using attentive models and lack the ability to exploit relational information such as user-product interactions or item co-occurrence relations. In this work, we propose to address the limitations of prior arts by exploring local and global user behavior patterns on a user successive behavior graph, which is constructed by utilizing short-term actions of all users. To capture implicit user preference signals and collaborative patterns, we use an efficient jumping graph convolution to explore high-order relations to enrich product representations for user preference modeling. Our approach can be seamlessly integrated with existing latent space based methods and be potentially applied in any product retrieval method that uses purchase history to model user preferences. Extensive experiments on eight Amazon benchmarks demonstrate the effectiveness and potential of our approach. The source code is available at
\url{https://github.com/floatSDSDS/SBG}.

\end{abstract}


\maketitle

\section{Introduction} \label{sec: intro}

Convenience drives the growth of e-commerce platforms such as Taobao
or Amazon.
Product search is an essential module in online shopping platforms, which guides users to browse and purchase products from a huge collection of commodities. Product search has its unique characteristics, making it distinct from web search, where information retrieval has made considerable progress.
    First, in web search engines, web pages are usually represented by long descriptive texts, while in e-commerce platforms, products are mainly represented by short texts such as titles and reviews, which may not always be informative.
    Second, other than textual representations, products are also associated with diverse relational data including ontology, spec sheet, figures, etc. 
    Third, there are various types of user-item interactions in e-commerce platforms. A user can browse, click, review, or purchase a product, or simply put it in his/her cart. Besides, there exist other structural information such as query-reformulation, shop browsing or category browsing, and shopping cart checkout. 
    
    It would be highly desirable yet challenging to utilize such rich information for personalized product search. Existing methods mainly exploit text data. Among them, a recent line of research~\cite{lse/conf/cikm/GyselRK16, hem/sigir/AiZBCC17, zam/cikm/AiHVC19, tem/conf/sigir/BiAC20} proposes to projects queries, items, and users into the same latent space and learn the representations of all entities with language modeling and information retrieval tasks, which enables the model to learn domain-specific semantic representations. However, they are limited in their ability to model user preferences, which is the core problem in product search. A common way to represent users is by the products they've visited during a period of time, 
    but long-term historical user behavior normally contains noisy preference signals.  HEM~\cite{hem/sigir/AiZBCC17} suffers from this problem since it represents a user with all his/her reviews of purchased products. 
    ZAM~\cite{zam/cikm/AiHVC19}, TEM~\cite{tem/conf/sigir/BiAC20}, and RTM~\cite{rtm/DBLP:conf/sigir/BiAC21} employ attentive models such as Transformer-based encoder to model user preferences and take into account both user behavior and query. For computational efficiency, user behavior sequences are usually truncated, and only recent behaviors are considered. While this helps to eliminate noisy preference signals, short-term user behavior may not contain sufficient preference signals (see more discussion in Sec.~\ref{sec:preliminaries}).

    To capture more useful user preference signals, it is a natural idea to explore various user-product interactions and product co-occurrence relationships, which are usually encoded in a graph. 
    Some recent efforts \cite{srrl:conf/cikm/LiuGCZ20, drem/journals/tois/AiZBC20} have been devoted to exploiting structural graph information for personalized product search. Ai et al.~\cite{drem/journals/tois/AiZBC20} proposed a dynamic relation embedding model (DREM). DREM constructs a unified knowledge graph to encode diverse relations and dynamic user-search/purchase behaviors and models the structural relationships via graph regularization. 
    Liu et al.~\cite{srrl:conf/cikm/LiuGCZ20} proposed graph embedding based structural relationship representation learning (GraphSRRL), which explicitly models the structural relationships, such as two users visiting the same product by a same query or a user visiting the same product by two different queries. 
    While DREM and GraphSRRL can model complex relationships, they include all previous user behaviors for preference modeling and may suffer from the noise induced by overly diverse signals. 

    In this work, we propose to explore local and global user behavior patterns on a user successive behavior graph (SBG) for user preference modeling. 
    The SBG is constructed by utilizing \emph{short-term} actions of all users, which collectively form a global behavior graph with rich relations among products. 
    To capture implicit user preference signals and collaborative patterns, we employ graph convolution to learn enriched product representations, which can be subsequently used for user preference modeling. Since user purchase behaviors are often sparse, it is helpful to explore high-order information on the SBG to model potential user interest, which requires stacking many graph convolution layers and leads to the well-known over-smoothing problem.
    To address this issue, we adopt an efficient jumping graph convolution layer that can effectively alleviate the over-smoothing effect. 
    To showcase the usefulness of our approach, we integrate it into a state-of-the-art latent space based model ZAM~\cite{zam/cikm/AiHVC19} and evaluate its performance on eight Amazon public benchmarks. The results show that our approach can significantly improve upon the base model and achieve better performance than other graph-based methods including DREM~\cite{drem/journals/tois/AiZBC20} and GraphSRRL~\cite{srrl:conf/cikm/LiuGCZ20}. It is worth noting that our approach is generic and can be potentially applied in any product retrieval method that models users using their purchase history.
    
    The contributions of this paper are summarized as follows.
    \begin{itemize}
    \item To our best knowledge, this is the first work to study how to improve product search with graph convolution, a technique that has recently been shown useful for many applications in various fields. 
   
    \item We propose to model successive user behavior and exploit local and global behavior patterns with graph convolution for user preference modeling. We also use an efficient graph convolution layer with jumping connections to alleviate the over-smoothing problem and theoretically analyze its effectiveness.  
    
    \item Extensive comparative experiments and ablation studies on eight Amazon benchmarks demonstrate the effectiveness of our proposed method, which can be potentially applied in any product retrieval method that models users with their purchase history for personalized product search. 
    \end{itemize}

\section{Related Work} \label{sec:related}

\subsection{Product Search}

\textbf{Latent spaced based product search.} \quad
In recent years, neural network based models have dominated the research of product search~\cite{guo2018multi, DBLP:conf/sigir/BiTDMC19, dbml/conf/cikm/XiaoRMSL19, alstp/guo2019attentive, robust/DBLP:conf/acl/NguyenRS20, click_through/DBLP:conf/www/YaoTCYXD021}. 
LSE is the first latent vector space based search framework proposed by Van Gysel et al.~\cite{lse/conf/cikm/GyselRK16}, which maps queries and products in the same space. Later on, Ai et al.~\cite{hem/sigir/AiZBCC17} proposed to consider user preferences and learn embeddings of users, products, and words jointly with two tasks: the language modeling task and the information retrieval task. Further, Ai et al.~\cite{zam/cikm/AiHVC19} considered user history behaviors conditioned on the current query and proposed a zero attention vector to control the degree of personalization.
In addition, in order to model long and short-term user preferences simultaneously, Guo et al.~\cite{alstp/guo2019attentive} designed a dual attention-based network to capture users' current search intentions and their long-term preferences. 
Recently, Transformer-based architecture have also been explored~\cite{tem/conf/sigir/BiAC20, rtm/DBLP:conf/sigir/BiAC21} for product search.
In this work, we employ the popular latent spaced based product search framework as in ZAM~\cite{zam/cikm/AiHVC19} and propose to enrich product representations via graph convolution.

\textbf{Graph-based product search.}
Early attempts employed relational information such as social signals or user behavior traces for search and ranking  problems~\cite{DBLP:journals/www/ChelaruOA14c, DBLP:conf/chiir/BadacheB17b, DBLP:conf/webi/Badache19a}. Recently, some studies attempted to model such information with graphs~\cite{geps/conf/www/ZhangWZ19, hypergraph_ps/DBLP:journals/mta/BuZQ20, drem/journals/tois/AiZBC20, srrl:conf/cikm/LiuGCZ20}. 
Zhang et al.~\cite{geps/conf/www/ZhangWZ19} proposed GEPS (Graph Embedding-based ranking model for Product Search) that employs pre-training techniques to learn product and query embeddings. As far as we know, it is the first attempt to use graphs for product search but it overlooks user preferences. 
Bu et al.~\cite{hypergraph_ps/DBLP:journals/mta/BuZQ20} proposed to model product textual semantic relationships with hypergraph to learn structural information. 
Ai et al.~\cite{drem/journals/tois/AiZBC20} proposed DREM (Dynamic Relation Embedding Model) that constructs a directed unified knowledge graph and jointly learns all embeddings through graph regularization. However, DREM lacks the ability to select informative information and is easily susceptible to noise. In this work, we exploit information on a successive behavior graph to avoid this problem. 
Liu et al.~\cite{srrl:conf/cikm/LiuGCZ20} proposed GraphSRRL (Graph embedding based Structural Relationship Representation Learning model) that explicitly utilizes specific user-query-product relationships. 
GraphSRRL pays attentions to local relations. In this work, we employ graph convolution with jump connections to aggregate high-order graph information.

\subsection{Information Retrieval with Graphs}

Graph-based methods have been extensively explored in the literature of sequential recommendation. These methods can be broadly categorized into two groups: embedding based methods and graph neural networks (GNN) based methods. Embedding-based methods~\cite{DBLP:conf/sigir/CraswellS07, DBLP:conf/sigir/GaoYLDN09, arci/conf/cikm/LiR0MM0M20, duong2020graph} employ network embedding techniques such as DeepWalk~\cite{deepwalk:conf/kdd/PerozziAS14} or Node2Vec~\cite{node2vec:conf/kdd/GroverL16}. They learn structural graph information by leveraging the skip-gram model~\cite{mikolov2013efficient}.
GNN based methods~\cite{arci/conf/cikm/LiR0MM0M20, niu2020dual, qi2020cgtr,  DBLP:conf/sigir/Wang0CLMQ20} employ GNN~\cite{DBLP:conf/nips/DefferrardBV16, kipf2016semi} to aggregate information over graphs. Our work is closer to GCE-GNN (Global Context Enhanced Graph Neural Networks)~\cite{DBLP:conf/sigir/Wang0CLMQ20} in that we both make use of global information. The difference is that GCE-GNN focuses more on transitions between items and utilizes the order of items, whereas we do not consider the order of products and pay more attention to their co-occurrence relationships. 

\section{Motivation and Insight} \label{sec:preliminaries}




In this section, we discuss the limitations of existing product search methods in user preference modeling and provide insight on our proposed approach.





\begin{figure}[t]
  \centering
  \includegraphics[width=\linewidth]{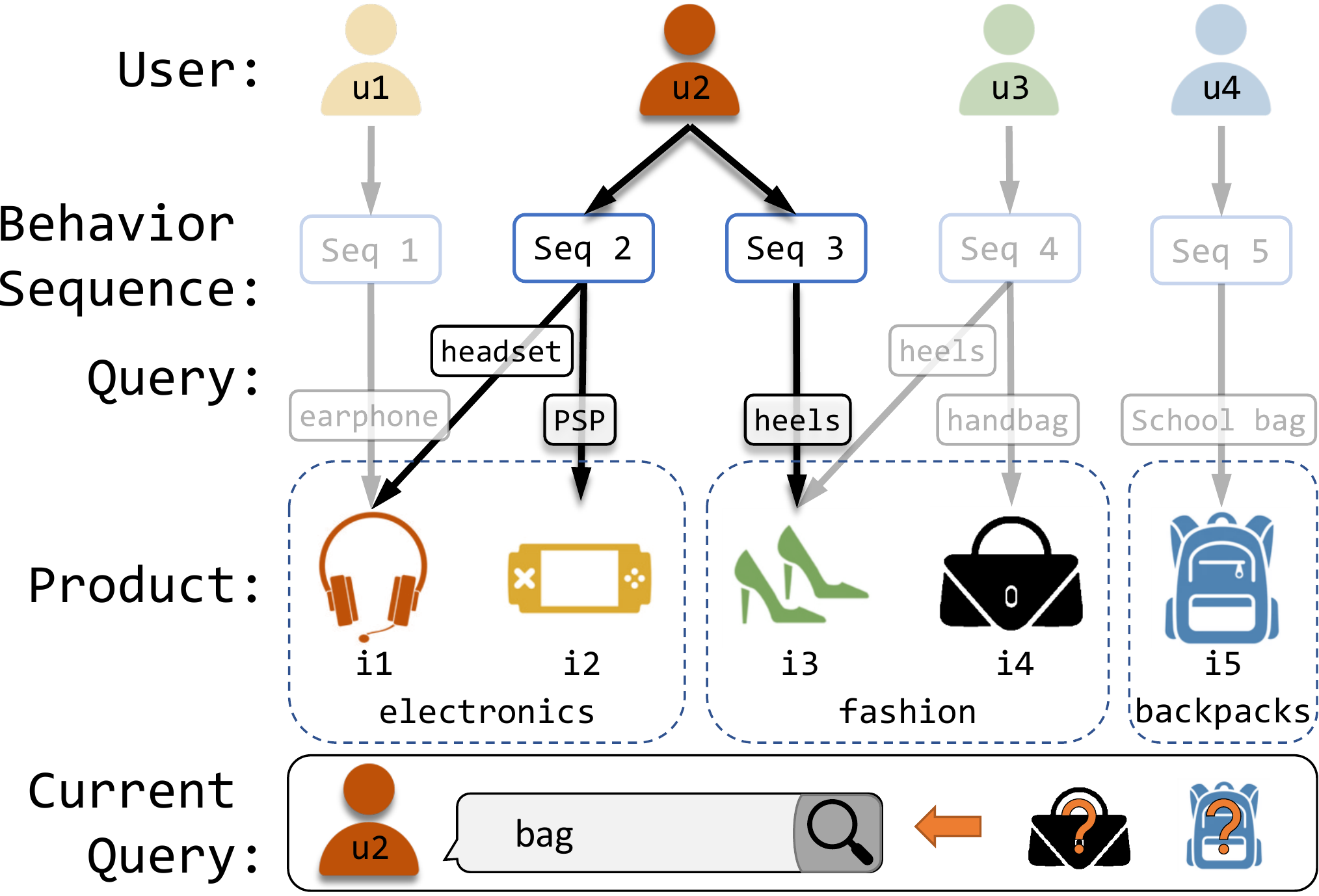}
  \caption{Illustration of exploiting global user successive behavior for personalized product search. 
  Given the query \textit{bag} issued by user $u_2$, the system is expected to retrieve suitable bags satisfying the user intent. By modeling the global user successive behavior graph with graph convolution, our proposed approach can capture implicit preference signals and yield desirable results.
  }
  \label{fig: motivation}
\end{figure}

\subsection{Limitations of Existing Methods}
Existing product search methods commonly model user preferences 
by considering their long-term or short-term behavior, but both of them have limitations.

\textbf{Long-term user behavior may contain noisy preference signals.} Typically, the long-term preference of a user is represented by all his/her historical interactions during a long period of time, which may contain many items and be overly diverse.
For example, as shown in Figure~\ref{fig: motivation}, the current query of user $u_2$ is ``\textit{bag}", but she has visited a couple of electronic devices in her purchase history, which are not related to bags. Therefore, representing $u_2$'s preference with all his/her historical engagement may impose negative effect on the current search induced by the irrelevant products in the purchase history. This problem is even more severe in HEM~\cite{hem/sigir/AiZBCC17}, which models users independently with their previous purchase reviews that may contain a lot noisy textual information.

\textbf{Short-term user behavior may not contain enough preference signals.} To address the above-mentioned issue, recent works including ZAM~\cite{zam/cikm/AiHVC19}, TEM~\cite{tem/conf/sigir/BiAC20} and RTM~\cite{rtm/DBLP:conf/sigir/BiAC21} only consider recent user actions in a short span of time and adopt attentive models such as Transformer-based encoder 
to put more emphasis on relevant products in the purchase record. 
However, a user's recent behavior may not contain enough preference signals for product retrieval. 
For example, when user $u_2$ searches for ``bag'', the remotely relevant item in his/her purchase history is ``heels'', which is not directly relevant to bags. Therefore, it is hard to tell whether the user is looking for handbags, backpacks, or other kinds of bags. 

\subsection{Insight of Our Proposed Approach}\label{sec:3_2_insight}
In this work, we propose to overcome 
the limitations of prior works in user preference modeling by exploring local and global user behavior patterns on a user successive behavior graph (SBG), which is constructed by utilizing \emph{short-term} actions of \emph{all}  users.
We then exploit \emph{high-order} relations in the SBG to capture implicit collaborative patterns and preference signals with an efficient jumping graph convolution and learn enriched product representations for user preference modeling. Our approach addresses the aforementioned problems in the following two aspects.

\textbf{Expanding the set of potentially intended products.} While short-term user behavior usually contains a limited number of products which may be inadequate to reflect user preferences, the \emph{global} SBG connects the products recently purchased by a user to other relevant or similar ones on the graph. In effect, it 
expands the set of potentially intended products of the user.
For example, in Figure~\ref{fig: motivation}, the heels $i_3$ is connected to the handbag $i_4$ because of the co-occurrence of $i_3$ and $i_4$ in the behavior sequence Seq 4 of user $u_3$. On the global SBG, the behavior sequences Seq 3 and Seq 4 are connected by $i_3$, and hence $i_4$ could be a potentially intended product for user $u_2$. 


\textbf{Making connected products more similar in the latent space.} 
Since a user is often represented as the ensemble of the products he/she bought, learning better product representations is crucial for user preference modeling. 
We leverage graph convolution to exploit the connectivity patterns in the SBG and make the embeddings of connected products more similar.
The enriched product representations can better reflect user preference. For example, by graph convolution, the fashion items $i_3$ and $i_4$ will be more similar. Therefore, when user $u_2$ searches for ``bag'', the handbag $i_4$ would be ranked higher than the backpack $i_5$, because $i_4$ is more similar to $i_3$ than $i_5$, and $i_3$ partially represents user $u_2$.

\section{Proposed Approach} \label{sec:method}

In this section, we present our proposed approach in detail,
which is built on the popular latent space based product search framework that learns semantic representations for products, users, and text tokens in the same embedding space. On top of it, we aim to learn enriched product representations to better represent users for search personalization. As shown in Figure~\ref{fig: framework}, we first construct a successive behavior graph from observed user behavior sequences. Then, we employ an efficient graph convolution with jumping connections to enrich product representations. The enriched product vectors can subsequently be used to represent users fro better preference modeling.  



\begin{figure}[t]
  \centering
  \includegraphics[width=.95\linewidth]{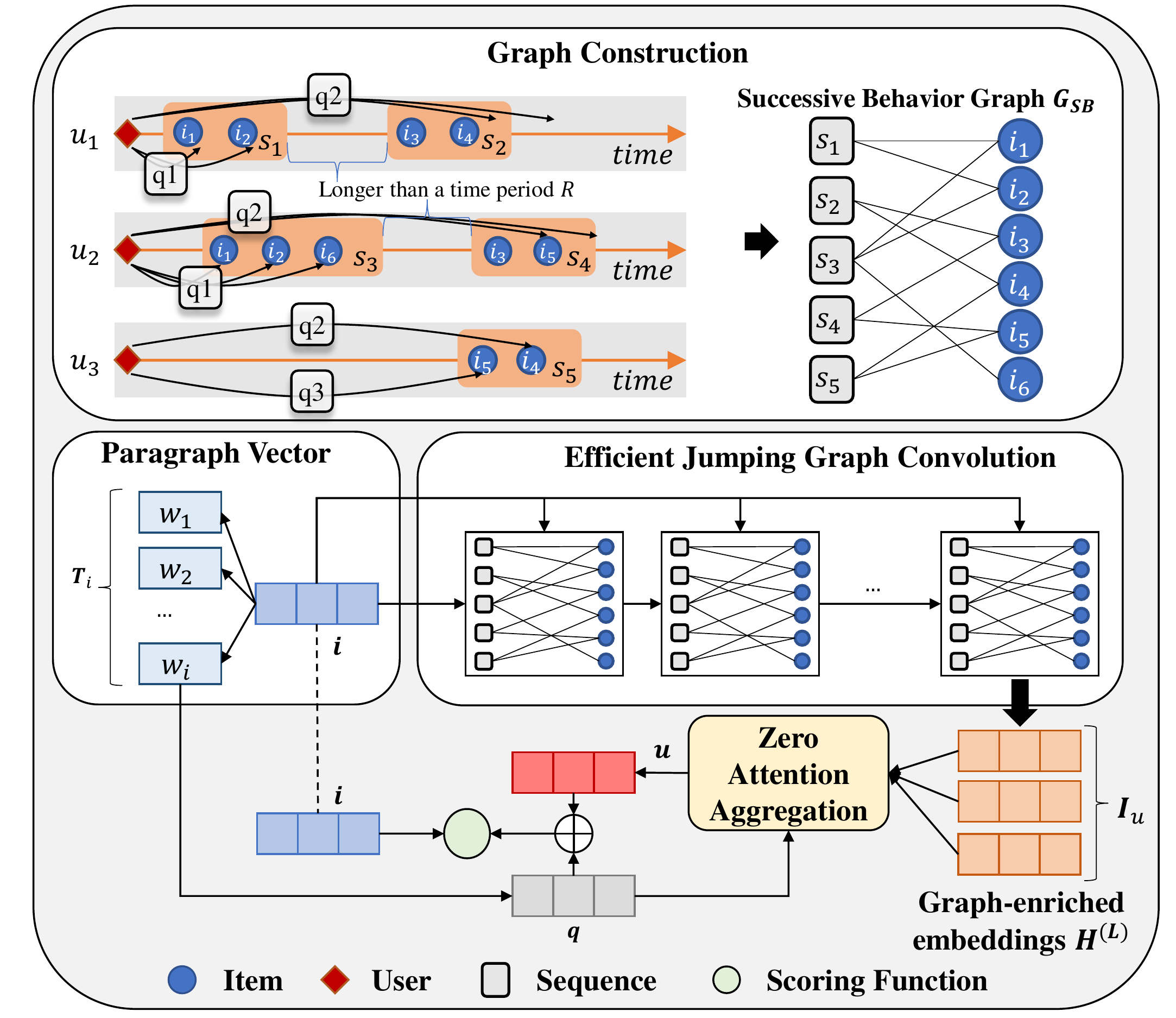}
  \caption{The framework of our proposed SBG model. 
  }
  \label{fig: framework}
\end{figure}


\subsection{Latent Space Based Product Search Framework} \label{sec:4_1_latent_space}

The goal of product search is to retrieve the most relevant products from a candidate pool $C$ for a user with his/her current query.
In this work, we adopt a typical latent space based generative framework for product search~\cite{lse/conf/cikm/GyselRK16, hem/sigir/AiZBCC17, zam/cikm/AiHVC19, tem/conf/sigir/BiAC20}, which learns embeddings for all products, users, and text tokens. The most relevant products are retrieved by matching product embeddings with the current context, which is usually represented by the user and query. 


More formally, after receiving a query $q$ from a user $u$, the system is expected to predict the probability $P(i|u, q)$ of user $u$ purchasing product $i$, for each product in the candidate set $C$, and then rank all the products by the probability. The available information includes the purchasing history of all the users and the text associated with the products, such as product name, product description, and product review. 
The latent space based product search framework learns entity embeddings from two tasks: the product retrieval task and the language modeling task.

\textbf{Product Retrieval Task.}\quad 
It aims to retrieve relevant products w.r.t. the current query. In \cite{hem/sigir/AiZBCC17}, the user intent $\bm {M}_{uq}$ is represented by a mix of the query $\bm q$ and the user's preference vector $\bm u$:  
\begin{equation}
    \bm {M}_{uq} = \lambda \bm {q} + (1 - \lambda) \bm {u},
\end{equation}
where $\lambda\in[0,1]$ is a balancing parameter. As such, $\bm {M}_{uq}$ encodes both the semantic meaning of the query and user preference. Then, the probability of purchasing product $i$ is computed as 
\begin{equation}
    P(i| u, q) = \frac{\exp(f(\bm {i}, \bm {M}_{uq}))}{\sum\limits_{i'\in C} \exp(f(\bm {i}', \bm {M}_{uq}))},
    \label{eq: p_i_uq}
\end{equation}
where $C$ is the set of all possible candidate products and $f$ is a similarity measure such as cosine similarity.

\textbf{Language Modeling Task.}\quad 
It aims to learn the embeddings of queries and products by modeling the text information. \cite{hem/sigir/AiZBCC17} proposes to jointly learn the word embedding $\bm w$ and product embedding $\bm i$ from the product's associated text by the paragraph vector (PV) model ~\cite{pv/le2014distributed}.
The PV model assumes that words or tokens can be generated from the entity 
and maximizes the likelihood 
\begin{equation}
    P(T_i | i)=\prod_{w\in T_i} \frac{\exp(\tau (w, i))}{\sum_{w'\in V} \exp(\tau(w', i))},
    \label{eq: p_w_i}
\end{equation}
where $\tau$ denotes a scoring function of product $i$ and its associated word $w$, and $T_i$ is the set of words associated with $i$. 

With the learned word embeddings, a query is then represented by a function of the word embeddings:
\begin{equation}
    \bm q = \phi(\{\bm w_q| w_q \in q\}),
\end{equation}
where $w_q$ is a word of query $q$, and $\phi$ 
can be a non-linear sequential encoder such as LSTM or Transformer. Since queries are usually short and the order of words often does not matter, we simply use an average function to obtain the query embedding.

\subsection{Efficient Graph Convolution with Jumping Connections} \label{sec:4_2/efficient_jumping}

As mentioned in Sec.~\ref{sec:preliminaries}, to improve user preference modeling and learn better product representations, we propose to utilize a global successive behavior graph and perform
graph convolution over the graph to capture implicit and complex collaborative signals. 
Here, we adopt an efficient graph convolution layer with jumping connections and provide theoretical analysis to show its advantage.

\textbf{Efficient Graph Convolution.}\quad 
In the past few years, graph convolutional networks (GCN) and variants  \cite{kipf2016semi,velickovic2017graph,hamilton2017inductive,chen2020simple/DBLP:conf/icml/ChenWHDL20} have been successfully applied to learn useful graph node representations for various graph learning and mining tasks.
In each layer of GCN, it performs feature propagation and transformation with connected nodes in the graph:
\begin{equation}\label{eq:convolution}
    \bm H^{(l)} = \sigma \left(\hat{\bm A} \bm H^{(l-1)} \bm W^{(l)}\right),
\end{equation}
where $\hat{\bm A}=\bm I+\bm D^{-1} \bm A$ is the (normalized) adjacency matrix with self-loops, and $\bm D$ is the degree matrix. $\bm H^{(l)}$ is the node embeddings produced by layer $l$. 
$\bm W^{(l)}$ denotes trainable parameters, and $\sigma$ is a non-linear function such as $\text{ReLU}(\cdot)$.




The projection layers (trainable parameters $\bm W^{(l)}$) and activation layers ($\sigma$) as shown in Eq.~(\ref{eq:convolution}) are commonly included in many GCN-based methods. However, as observed from our empirical study, the projection layers may distort the semantic product representations learned by language modeling in methods such as ZAM or HEM. Hence, we propose to use graph convolution without the projection layers to enrich product representations. Following the efficient design in Li et al.~\cite{li2019label}, we remove the projection layers and activation layers. Further, we add a balancing parameter $\omega$ to control the strength of self-information: 
\begin{equation}\label{eq:simple}
    \bm H^{(l)}=\left(\omega\bm I + (1-\omega)\bm D^{-1}\bm A\right) \bm H^{(l-1)}.
\end{equation}
\textbf{Jumping Graph Convolution Layer.}\quad 
Since user purchase behavior is often sparse, it is helpful to aggregate high-order information on the successive behavior graph to model potential user interest, as discussed in Sec.~\ref{sec:3_2_insight}. However, the ordinary graph convolution suffers from the well-known over-smoothing problem~\cite{dpgcn/DBLP:conf/aaai/LiHW18}, i.e., stacking too many convolution layers may make the node features (product representations) indistinguishable. To address this issue, Chen et al.~\cite{chen2020simple/DBLP:conf/icml/ChenWHDL20} proposed GCNII that adds initial residual connections to each GCN layer. We follow the same design to add jumping connections, i.e., feeding each convolution layer an additional input of the initial product representations $\bm H^{(0)}$:
\begin{equation}
    \tilde{\bm H}^{(l)}=\left(\omega\bm I + (1-\omega)\bm D^{-1}\bm A\right) \left (\beta \bm H^{(0)} + (1-\beta)\tilde{\bm H}^{(l-1) } \right)\label{eq:jump_enriched_product},
\end{equation}
where $\beta$ is a weight parameter determining the portion of initial features. Our experiments in Sec.~\ref{sec:5_3_2_effect_jumping} verify the effectiveness of jumping connections, which can alleviate the over-smoothing effect and enable utilizing high-order relations on the graph.

Further, we provide a theoretical analysis of jumping group convolution by measuring the diversity of product representations after graph convolution using Laplacian-Beltrami operator $\Omega(\cdot)$ \cite{chung1997spectral}. We compare the diversity of product representations with jumping connections ($\tilde{\bm H}^{(l)}$ in Eq.~(\ref{eq:jump_enriched_product})) and those without jumping connections ($\bm H^{(l)}$ in Eq.~(\ref{eq:simple})). We employ Laplacian-Beltrami operator, which measures the total variance of connected nodes:
\begin{equation}\label{eq:laplacian}
    \Omega(\bm H) = \sum_k \sum_{i,j} a_{ij} (\bm H_{i,k} - \bm H_{j,k})^2.
\end{equation}
High $\Omega(\bm H)$ indicates high diversity, and low $\Omega(\bm H)$ indicates severe over-smoothing. The following theorem shows jumping connections can substantially alleviate the over-smoothing effect of graph convolution.

\begin{thm}\label{thm:jump}
If the initial diversity $\Omega(\bm H^{(0)}) >0$, then for any integer $l > 0$, $\beta \in (0,1)$, and $\omega \in (0.5, 1)$, $\tilde{\bm H}^{(l)}$ is strictly more diverse than $\bm H^{(l)}$, i.e.,
\begin{equation}
  \Omega\left(\tilde{\bm H}^{(l)}\right) > \Omega\left({\bm H}^{(l)}\right). 
\end{equation}
When $l$ approaches infinity, jumping connections can prevent the diversity of product representations from collapsing to 0 (over-smoothing), i.e.,
\begin{equation}
    \lim_{l\to\infty} \Omega\left(\tilde{\bm H}^{(l)}\right) > \lim_{l\to\infty}\Omega\left(\bm H^{(l)}\right) = 0.
\end{equation}
\end{thm}
\begin{proof}
The proof is provided in the Appendix.
\end{proof}

\subsection{Modeling User Behavior with Graph Convolution}

Here, we show how to utilize local and global user behavior patterns to improve user representations in a zero attention model with the efficient jumping graph convolution. First, we construct a bipartite successive behavior graph. Then, we stack multiple jumping graph convolution layers to enrich product representations. Finally, we incorporate the graph-enriched product vectors into the zero attention model for user preference modeling. 



\textbf{Graph Construction.}\quad To construct the successive behavior graph, we first define successive behavior sequences in the training set. First, we sort all the observed purchased records in a chronological order for each user. If the time interval between two consecutive actions is within a period $R$ (e.g., a day, a week, or a month), the two actions are considered as successive and will be placed in the same successive behavior sequence.
Then, we construct a successive behavior graph $G_{SB}$, which is a bipartite graph between sequences and products. If and only if a product $i$ is in a sequence $S$, we form an edge between $i$ and $S$, denoted as $G_{SB}(i, S)=1$.


\textbf{Enriching Product Representations with Graph Convolution.} \quad
To enrich product representations, we apply jumping graph convolution as introduced in Sec.~\ref{sec:4_2/efficient_jumping} on the successive behavior graph $G_{SB}$. Let $\bm h^{(0)}_j$ denote the input of the first graph convolution layer for any entity $j$. We use the embeddings learned with PV as $\bm h^{(0)}_i$ for a product $i$. For a sequence $s$, $\bm h^{(0)}_s$ is randomly initialized. We then apply $L$ efficient jumping graph convolution layers as defined in Eq.~(\ref{eq:jump_enriched_product}) and obtain the graph-enriched product embedding $\tilde{\bm h}^{(L)}_i$ for each product. 


\textbf{Using Graph-enriched Product Representations for User Preference Modeling.} \quad
Based on the observation that the effect of personalization varies significantly in respect of query characteristics. Ai et al.~\cite{zam/cikm/AiHVC19} proposed ZAM that introduces a zero-vector to adaptively control the degree of personalization.
The representation of a user $u$ is composed of his/her recently visited products, which is computed as 
\begin{equation}
    \bm {u} = 
    \sum_{i\in I_u \cup \bm {0}} 
    \frac{\exp(s(q, i))}{\exp(s(q, \bm {0})) + \sum_{i'\in I_u} \exp(s(q, i'))}\bm i,
    \label{eq: zam}
\end{equation}
where $I_u$ is the product set in user $u$'s history visit records, $\bm {0}$ is a zero vector. The attention score $s(q, i)$ of a given product $i$ w.r.t. the current query $q$ is defined as
\begin{equation}
    s(q, i) = (\bm {i}^\top \tanh(\bm {W}_f^\top\bm {q} + \bm {b}_f))^\top\bm {W}_h,
    \label{eq: zam_s_qi}
\end{equation}
where $\bm {W}_h \in \mathbb{R}^{d_a}$, $\bm {W}_f \in \mathbb{R}^{d \times d_a \times d}$, $\bm {b}_f \in \mathbb{R}^{d \times d_a}$ are the trainable parameters, and $d_a$ is the hidden dimension of the user-product attention network. In particular, $exp(s(q, \bm 0))$ is calculated by Eq.~(\ref{eq: zam_s_qi}) with $\bm i$ as a learnable inquiry vector $\bm {0'}\in \mathbb{R}^d$.



To incorporate the graph-enriched product representations for user preference modeling, we use $\tilde{\bm h}_i^{(L)}$ to substitute $\bm {i}$ in Eq.~(\ref{eq: zam}) and Eq.~(\ref{eq: zam_s_qi}), i.e.,
\begin{equation}
\begin{aligned}
    & \bm {u} = 
    \sum_{i\in I_u \cup \bm {0}} 
    \frac{\exp(s(q, i))}{\exp(s(q, \bm {0})) + \sum_{i'\in I_u} \exp(s(q, i'))}\tilde{\bm h}_i^{(L)}, \\
    & s(q, i) = (\tilde{\bm {h}}_i^\top \tanh(\bm {W}_f^\top\bm {q} + \bm     {b}_f))^\top\bm {W}_h.
    \label{zam_substitute}    
\end{aligned}
\end{equation}





\subsection{Model Optimization}

Following ZAM~\cite{zam/cikm/AiHVC19}, we jointly optimize the product retrieval task and the language modeling task.
The product retrieval loss is 
\begin{equation}
    L_{PR} = -\sum_{(u,i,q)}\log P(i|u,q) 
        = -\sum_{(u, i, q)}\log \frac{\exp(f(\bm {i}, \bm {M}_{uq}))}{\sum\limits_{i'\in C} \exp(f(\bm {i}', \bm {M}_{uq}))},
\label{eq: loss_pr_uqi}
\end{equation}
which is optimized over all triples $(u, i, q)$ in the training set, where a triple $(u, i, q)$ represents a product $i$ purchased by a user $u$ under the submitted query $q$. The language modeling loss is
\begin{equation}
    L_{LM} = -\sum_i \log P(T_i | i) 
        =  -\sum_{i}\sum_{w\in T_i}\log  \frac{\exp(\tau (w, i))}{\sum\limits_{w'\in V} \exp(\tau(w', i))}.
\label{eq: loss_lm_iti}
\end{equation}
Hence, the total loss is
\begin{equation}
    L_{total} = L_{PR} + L_{LM} \\
        =  -\sum_{(u, i, q)}\log P(i | u, q) -\sum_{i}\log P(T_i| i). 
\end{equation}

\textbf{Remark.} It is worth noting that we only use the graph-enriched product embeddings to represent users but do not use them to represent products themselves in the product retrieval task (Eq.~(\ref{eq: p_i_uq})) or the language modeling task (Eq.~(\ref{eq: p_w_i})), because the mixed representations may make products lose their uniqueness and hurt performance, which is verified by our empirical study.

Since the candidate set $C$ 
and vocabulary $V$ 
are usually extremely large, it is impractical to compute the log likelihood in Eq.~(\ref{eq: loss_pr_uqi}) and Eq.~(\ref{eq: loss_lm_iti}). A common solution 
is to sample only a portion of negative products to approximate the denominator of Eq.~(\ref{eq: loss_pr_uqi}) and Eq.~(\ref{eq: loss_lm_iti}):
\begin{equation}
\begin{aligned}
    L=
        & -\sum_{(u, i, q)} [\log \tau (w, i) + k_w \mathbb{E}_{w' \sim P_w} \log \tau (-w', i) \\
        & + \log f(i, M_{uq}) + k_i \mathbb{E}_{i' \sim P_i}(\log f(-i, M_{uq}))],
\end{aligned}
\label{eq: loss_general}
\end{equation}
where $k_w$ and $k_i$ are the negative sampling rates for words and products, respectively. In this work, 
we follow HEM~\cite{hem/sigir/AiZBCC17} to randomly sample negative words from the vocabulary with $P_w$ as the unigram distribution raised to the 3/4rd power, and randomly sample negative products from all products in the training set.

\section{Experiments} \label{sec:experiment}

In this section, we introduce our experimental settings and present experimental results. Our experiments try to answer the following research questions:

\begin{itemize}
    \item \textbf{RQ1:} How is the performances of SBG compared to the base model ZAM? 
    \item \textbf{RQ2:} How does SBG perform compared to state-of-the-art methods for personalized product search?
    \item \textbf{RQ3:} How useful is graph convolution? Can the proposed jumping connection alleviate over-smoothing?
    \item \textbf{RQ4:} What is the effect of time interval $R$ on the constructed successive behavior graph $G_{SB}$?
\end{itemize}


\subsection{Experimental Setup}

\subsubsection{Datasets} 

\begin{table}[]
\small
\caption{Dataset Statistics.}
\scalebox{0.9}{
\begin{tabular}{lllll}
\toprule \hline
\textbf{}           & \textbf{Magazine}    & \textbf{Software} & \textbf{Phones} & \textbf{Toys\&Games}   \\ \hline
\textbf{\# reviews} & 4,583                & 25,086            & 133,792         & 148,756                \\
\textbf{\#user}     & 694                  & 3,642             & 17,464          & 16,370                 \\
\textbf{\#query}    & 170                  & 999               & 163             & 399                    \\
\textbf{\#product}  & 876                  & 5,875             & 10,278          & 11,875                 \\
\textbf{\#seq}      & 2,337                 & 17,814             & 79,224           & 78,616                  \\
\textbf{\#edge}     & 3,078                & 16,391            & 93,174          & 111,578                \\
\hline \hline
\textbf{}           & \textbf{Instruments} & \textbf{Clothing} & \textbf{Health} & \textbf{Home\&Kitchen} \\ \hline
\textbf{\# reviews} & 209,229              & 270,854           & 334,025         & 545,083                \\
\textbf{\#user}     & 23,887               & 37,914            & 36,639          & 65,510                 \\
\textbf{\#query}    & 492                  & 2,000             & 793             & 900                    \\
\textbf{\#product}  & 9,756                & 23,033            & 17,956          & 27,888                 \\
\textbf{\#seq}               &  100,945               & 160,959           & 201,513         & 333,709                \\
\textbf{\#edge}              &  149,401              & 189,985           & 256,095         & 408,607               \\
\hline \bottomrule
\end{tabular}}
\label{tb:data}
\end{table}

Our experiments are conducted on the well-known Amazon review dataset\footnote{http://jmcauley.ucsd.edu/data/amazon}~\cite{amazon/ni2019justifying}, 
which includes product reviews and metadata such as product titles and categories. 
It was first introduced for product search by Van Gysel et al.~\cite{lse/conf/cikm/GyselRK16, VanGysel2017sert} and has become a benchmark dataset for evaluating product search methods as used in many recent studies~\cite{hem/sigir/AiZBCC17, zam/cikm/AiHVC19, drem/journals/tois/AiZBC20, srrl:conf/cikm/LiuGCZ20}. Product reviews are generally used as text corpus for representing products or users, and product categories are used as queries to simulate a search scenario. 

In our experiments, we use the 5-core data, and for each sub dataset, we filter out products with no positive records. 
The statistics of the filtered datasets are shown in Table~\ref{tb:data}.




\subsubsection{Baselines} 

We compare our proposed approach SBG with the following baselines. 

\begin{itemize}
    \item \textbf{HEM~\cite{hem/sigir/AiZBCC17}} assumes that users and products are independent. It employs PV to learn the representations of users, products, and words jointly. The user-query pair is projected to the same latent space with products. 
    \item \textbf{ZAM~\cite{zam/cikm/AiHVC19}} also employs PV to learn semantic representations of products and words. Users are represented by the products they visited. In particular, ZAM proposes a zero attention vector to control the degree of personalization.
    \item \textbf{DREM~\cite{drem/journals/tois/AiZBC20}} employs knowledge graph embedding techniques and constructs a unified knowledge graph that represents both static entity features and dynamic user searching behaviors. Embeddings of all entities are learned via a graph regularization loss.
    \item \textbf{GraphSRRL~\cite{srrl:conf/cikm/LiuGCZ20}} explicitly utilizes structural patterns in a user-query-product graph. It defines three specific structural patterns that represent three frequent user-query-product interactions.
\end{itemize}

\begin{table*}[]
\small
\caption{Comparison of our proposed method with baselines on eight Amazon review sub datasets. $^*$ and $\dagger$ denote the significant differences to ZAM and the best baseline, respectively in paired t-test with $p \le 0.01$. The best results are highlighted in bold.}
\scalebox{0.9}{
\begin{tabular}{l|ccccc|ccccc}
\toprule
\hline
                    & \multicolumn{5}{c|}{\textbf{Software}}                                                            & \multicolumn{5}{c}{\textbf{Magazine}}                                                             \\
                    & \textbf{HR@10}    & \textbf{MRR@100}  & \textbf{NDCG@10}  & \textbf{NDCG@20}  & \textbf{NDCG@100} & \textbf{HR@10}    & \textbf{MRR@100}  & \textbf{NDCG@10}  & \textbf{NDCG@20}  & \textbf{NDCG@100} \\ \hline
\textbf{HEM}        & 0.3785            & 0.2181            & 0.2474            & 0.2688            & 0.3024            & 0.4115            & 0.2183            & 0.2526            & 0.2837            & 0.3211            \\
\textbf{ZAM}        & 0.4841            & 0.2900            & 0.3289            & 0.3465            & 0.3713            & 0.4349            & 0.2290            & 0.2657            & 0.2968            & 0.3384            \\
\textbf{DREM}       & 0.5058            & 0.3189            & 0.3555            & 0.3740            & 0.4029            & 0.3985            & 0.2130            & 0.2469            & 0.2725            & 0.3103            \\
\textbf{GraphSRRL}  & 0.2555            & 0.1415            & 0.1598            & 0.1795            & 0.2186            & 0.2682            & 0.1353            & 0.1579            & 0.1795            & 0.2190            \\
\textbf{SBG (ours)} & \textbf{0.5629*†} & \textbf{0.3759*†} & \textbf{0.4144*†} & \textbf{0.4302*†} & \textbf{0.4501*†} & \textbf{0.4679*†} & \textbf{0.2568*†} & \textbf{0.2952*†} & \textbf{0.3255*†} & \textbf{0.3657*†} \\ \hline
                    & \multicolumn{5}{c|}{\textbf{Phones}}                                                              & \multicolumn{5}{c}{\textbf{Toys\&Games}}                                                                 \\
                    & \textbf{HR@10}    & \textbf{MRR@100}  & \textbf{NDCG@10}  & \textbf{NDCG@20}  & \textbf{NDCG@100} & \textbf{HR@10}    & \textbf{MRR@100}  & \textbf{NDCG@10}  & \textbf{NDCG@20}  & \textbf{NDCG@100} \\ \hline
\textbf{HEM}        & 0.4049            & 0.2411            & 0.2697            & 0.2913            & 0.3314            & 0.2509            & 0.1330            & 0.1501            & 0.1729            & 0.2217            \\
\textbf{ZAM}        & 0.5160            & 0.2995            & 0.3386            & 0.3659            & 0.4083            & 0.4358            & 0.2314            & 0.2653            & 0.2985            & 0.3541            \\
\textbf{DREM}       & 0.5836            & 0.3365            & 0.3841            & 0.4107            & 0.4419            & 0.5557            & 0.3124            & 0.3554            & 0.3916            & 0.4368            \\
\textbf{GraphSRRL}  & 0.1883            & 0.0887            & 0.1035            & 0.1225            & 0.1581            & 0.3811            & 0.1822            & 0.2166            & 0.2475            & 0.2956            \\
\textbf{SBG (ours)} & \textbf{0.5878*}  & \textbf{0.3447*†} & \textbf{0.3904*}  & \textbf{0.4173*}  & \textbf{0.4553*†} & \textbf{0.571*†}  & \textbf{0.3182*}  & \textbf{0.3648*†} & \textbf{0.3976*}  & \textbf{0.4418*}  \\ \hline
                    & \multicolumn{5}{c|}{\textbf{Instruments}}                                                              & \multicolumn{5}{c}{\textbf{Clothing}}                                                        \\
                    & \textbf{HR@10}    & \textbf{MRR@100}  & \textbf{NDCG@10}  & \textbf{NDCG@20}  & \textbf{NDCG@100} & \textbf{HR@10}    & \textbf{MRR@100}  & \textbf{NDCG@10}  & \textbf{NDCG@20}  & \textbf{NDCG@100} \\ \hline
\textbf{HEM}        & 0.4754            & 0.2562            & 0.2974            & 0.3262            & 0.3649            & 0.5146            & 0.3013            & 0.3409            & 0.3673            & 0.4065            \\
\textbf{ZAM}        & 0.4951            & 0.2806            & 0.3202            & 0.3482            & 0.3911            & 0.5606            & 0.3230            & 0.3671            & 0.3972            & 0.4371            \\
\textbf{DREM}       & 0.4503            & 0.2475            & 0.2856            & 0.3107            & 0.3523            & 0.4129            & 0.2460            & 0.2741            & 0.2999            & 0.3488            \\
\textbf{GraphSRRL}  & 0.5073            & 0.2770            & 0.3194            & 0.3507            & 0.3951            & 0.1789            & 0.0811            & 0.0946            & 0.1152            & 0.1537            \\
\textbf{SBG (ours)} & \textbf{0.5184*†} & \textbf{0.3052*†} & \textbf{0.3451*†} & \textbf{0.3713*†} & \textbf{0.4116*†} & \textbf{0.602*†}  & \textbf{0.3528*†} & \textbf{0.4001*†} & \textbf{0.4294*†} & \textbf{0.4663*†} \\ \hline
                    & \multicolumn{5}{c|}{\textbf{Health}}                                                              & \multicolumn{5}{c}{\textbf{Home\&Kitchen}}                                                        \\
                    & \textbf{HR@10}    & \textbf{MRR@100}  & \textbf{NDCG@10}  & \textbf{NDCG@20}  & \textbf{NDCG@100} & \textbf{HR@10}    & \textbf{MRR@100}  & \textbf{NDCG@10}  & \textbf{NDCG@20}  & \textbf{NDCG@100} \\ \hline
\textbf{HEM}        & 0.3841            & 0.2292            & 0.2572            & 0.2778            & 0.3153            & 0.4507            & 0.2846            & 0.3164            & 0.3349            & 0.3666            \\
\textbf{ZAM}        & 0.4528            & 0.2622            & 0.2952            & 0.3235            & 0.3727            & 0.5247            & 0.3219            & 0.3593            & 0.3853            & 0.4275            \\
\textbf{DREM}       & 0.5667            & 0.3620            & 0.3985            & 0.4276            & 0.4686            & 0.5793            & 0.3894            & 0.4236            & 0.4501            & 0.4904            \\
\textbf{GraphSRRL}  & 0.5073            & 0.2770            & 0.3194            & 0.3507            & 0.3951            & 0.4534            & 0.2199            & 0.2604            & 0.2976            & 0.3534            \\
\textbf{SBG (ours)} & \textbf{0.6181*†} & \textbf{0.3696*†} & \textbf{0.4174*†} & \textbf{0.4458*†} & \textbf{0.4815*†} & \textbf{0.6419*†} & \textbf{0.4095*†} & \textbf{0.4546*†} & \textbf{0.4802*†} & \textbf{0.5152*†} \\ \hline \bottomrule
\end{tabular}}

\label{tb: amazon_8}
\end{table*}

\subsubsection{Evaluation Protocol} 

We partition each sub dataset into a training set, a validation set, and a test set. Similar to the process of constructing $G_{SB}$, we sort the reviews in chronological order for each user and split the full record into successive behavior sequences. Then, the last sequence is used as the test set, and the second last sequence is used as the validation set.  


For performance measurement, we adopt three metrics: hit rate (HR@K), normalized discounted cumulative gain (NDCG@K), and mean reciprocal rank (MRR). For each sequence, the target products are mixed up with candidates randomly sampled from the entire product set, forming a candidate set of size 1000.

\subsubsection{Implementation Details} 

For a fair comparison, we re-implement HEM, ZAM, and DREM using the same encoder, evaluation setting, and negative sampling method. For DREM, we build a unified heterogeneous graph that contains the user-product review relation and the product-category belonging relation. In addition, we connect products and users with their associated words during training. For GraphSRRL, we use the official implementation\footnote{https://github.com/Shawn-hub-hit/GraphSRRL-master} with our dataset splits and evaluation protocols. 
For all methods, the batch size is set to 1024, and the ADAM optimizer is used with an initial learning rate of 0.001. All the entity embeddings are initialized randomly with dimension 64. 
For our SBG, we set the attention dimension $d_a$ to 8, and the user-query balancing parameter $\lambda$ to 0.5. We employ 4 layers of jumping graph convolution, and the weight of self-loop is set to 0.1. The strength of jumping connection $\beta$ is also set to 0.1. The negative sampling rate for each word is set to 5, and that for each item is set to 2. We report the evaluation metrics on the converged model. The reported results are averaged over multiple runs, and the significant differences are computed based on the paired t-test with $p \le 0.01$.

\subsection{Main Results}

\begin{table}[]
\small
\caption{The improvement percentages of NDCG@10 over ZAM by DREM (the best baseline) and our SBG.}
\scalebox{0.9}{
\begin{tabular}{lllll}
\toprule
\hline
                    & \textbf{Magazine}    & \textbf{Software} & \textbf{Phones} & \textbf{Toys\&Games}   \\ \hline
\textbf{DREM}       & +8.08\%              & -7.06\%           & +13.44\%         & +33.96\%                \\
\textbf{SBG} & +25.99\%              & +11.10\%           & +15.30\%         & +37.50\%                \\ \hline
\textbf{}           & \textbf{Instruments} & \textbf{Clothing} & \textbf{Health} & \textbf{Home\&Kitchen} \\ \hline
\textbf{DREM}       & -10.83\%              & -25.34\%          & +35.00\%         & +17.89\%                \\
\textbf{SBG} & +7.76\%              & +8.99\%            & +41.39\%         & +26.53\%                \\ \hline \bottomrule
    \end{tabular}}
\label{tb:gain}
    \end{table}

Table~\ref{tb: amazon_8}
summarizes the overall performance of our SBG and the baselines on eight Amazon review sub datasets. Besides, the improvement percentages of SBG and the best baseline DREM (also graph-based) over ZAM are summarized in Table~\ref{tb:gain}. We can make the following observations.

\begin{itemize}

    \item First of all, SBG significantly improves over ZAM in every tested domain/dataset. Since NDCG reflects the quality of the entire ranking list, we calculate the performance gain of SBG compared to ZAM in NDCG@10 as shown in Table~\ref{tb:gain}. SBG improves ZAM by at least 7.76\% on \textit{Instruments}, and up to 41.39\% on \textit{Health}.
    \textbf{RQ1} is answered. 
    
    \item SBG achieves significant improvements over other baselines in nearly all cases, which answers \textbf{RQ2}. 
    

    
    \item 
    As shown in Table~\ref{tb:gain}, 
    the performances of DREM and SBG are consistent on all domains. In the domains where DREM fails including \textit{Software}, \textit{Instruments}, and \textit{Clothing}, the improvement percentages of SBG are also less significant. It indicates that the effectiveness of graph-based methods may be affected by the characteristics of different domains.  

    \item GraphSRRL assumes that the pre-defined patterns are frequent, which may not hold in some domains, as evidenced by the experimental results. 
\end{itemize}



\begin{figure}[t]
  \centering
    \includegraphics[width=0.8\linewidth]{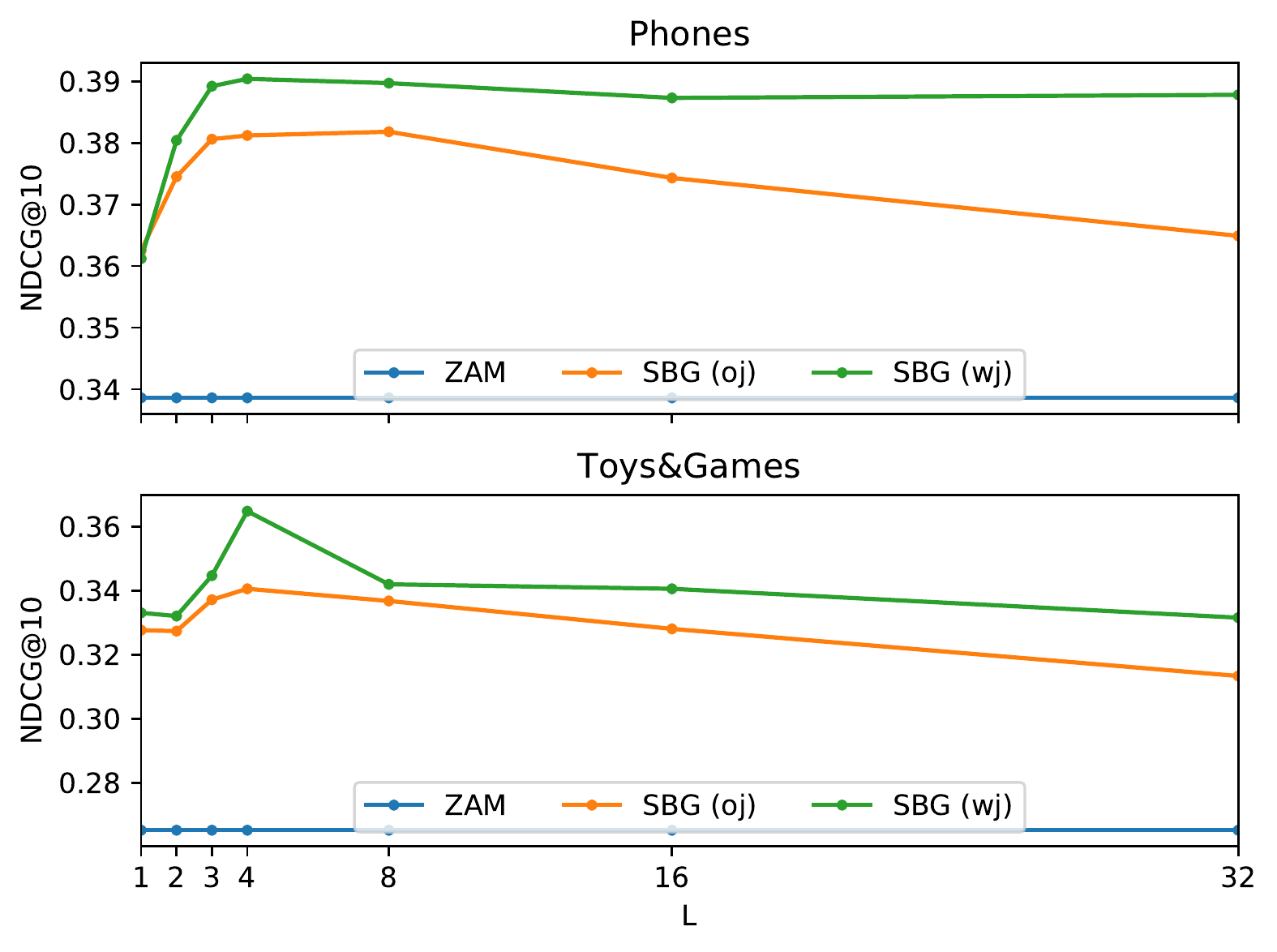}
  \caption{Ablation studies of jumping connections and the number $L$ of the graph convolutional layers. The results in NDCG@10 with respect to $L$ on \textit{Toys\&Games} and \textit{Phones} are reported. 
  SBG (wj) stands for our proposed SBG with jumping connections, and SBG (oj) stands for SBG without jumping connections.}
  \label{fig: l}
\end{figure}

\subsection{Analysis and Discussion}

\subsubsection{\textbf{Effect of the Order of Graph Convolution}}
To answer \textbf{RQ3}, we investigate the effect of the number of graph convolution layers $L$ and the jumping connections. We conduct experiments on \textit{Phones} and \textit{Toys\&Games} and vary $L$ from 0 to 64. 
We compare our SBG with jumping connections (denoted as SBG (wj)) with ZAM and a variant SBG (oj), which stands for SBG without jumping connections. 
The results are shown in Figure~\ref{fig: l}.

It can be seen that a few graph convolution layers can bring substantial performance gains. In our experiments on \textit{Phones} and \textit{Toys\&Games}, we achieve the best performance with $L=4$.   
However, as $L$ increases, the performances of SBG (wj) and SBG (oj) drop due to the over-smoothing effect. Especially on \textit{Toys\&Games}, we observe a significant performance drop when $L$ is larger than 4. 


\subsubsection{\textbf{Effect of Jumping Connections.}} \label{sec:5_3_2_effect_jumping}
It can be seen from Figure~\ref{fig: l} that as $L$ increases, the performance of SBG (wj) drops much slower than that of SBG (oj), especially on \textit{Phones}, demonstrating the effectiveness of jumping connections in alleviating the over-smoothing effect. 

\subsubsection{\textbf{Effect of Time Interval $R$.}}

To answer \textbf{RQ4},
we evaluate the performance of SBG with respect to different time scale $R\in \{a\ day, \\ a\ week, a\ month, a\ quater, a\ year\}$. The results shown in Figure~\ref{fig: time_scale}. It can be observed that the best performance is usually achieved when $R$ is $a\ day$ or $a\ week$. When $R$ is $a\ month$ or longer, the performances often drop. 
This is probably because the behavior sequences are overly diverse and contain noisy preference signals when the time scale is large.
The only exception is on \textit{Magazine}, where the performance of SBG slightly increases as $R$ becomes larger. However, \textit{Magazine} is a small dataset that does not have much training data and many products have limited records in the test set. Hence, content-based product features may not be informative enough, and a larger $R$ helps a cold product reach to broader neighborhood and leads to better performance.

\begin{figure}[t]
  \centering
    \includegraphics[width=1\linewidth]{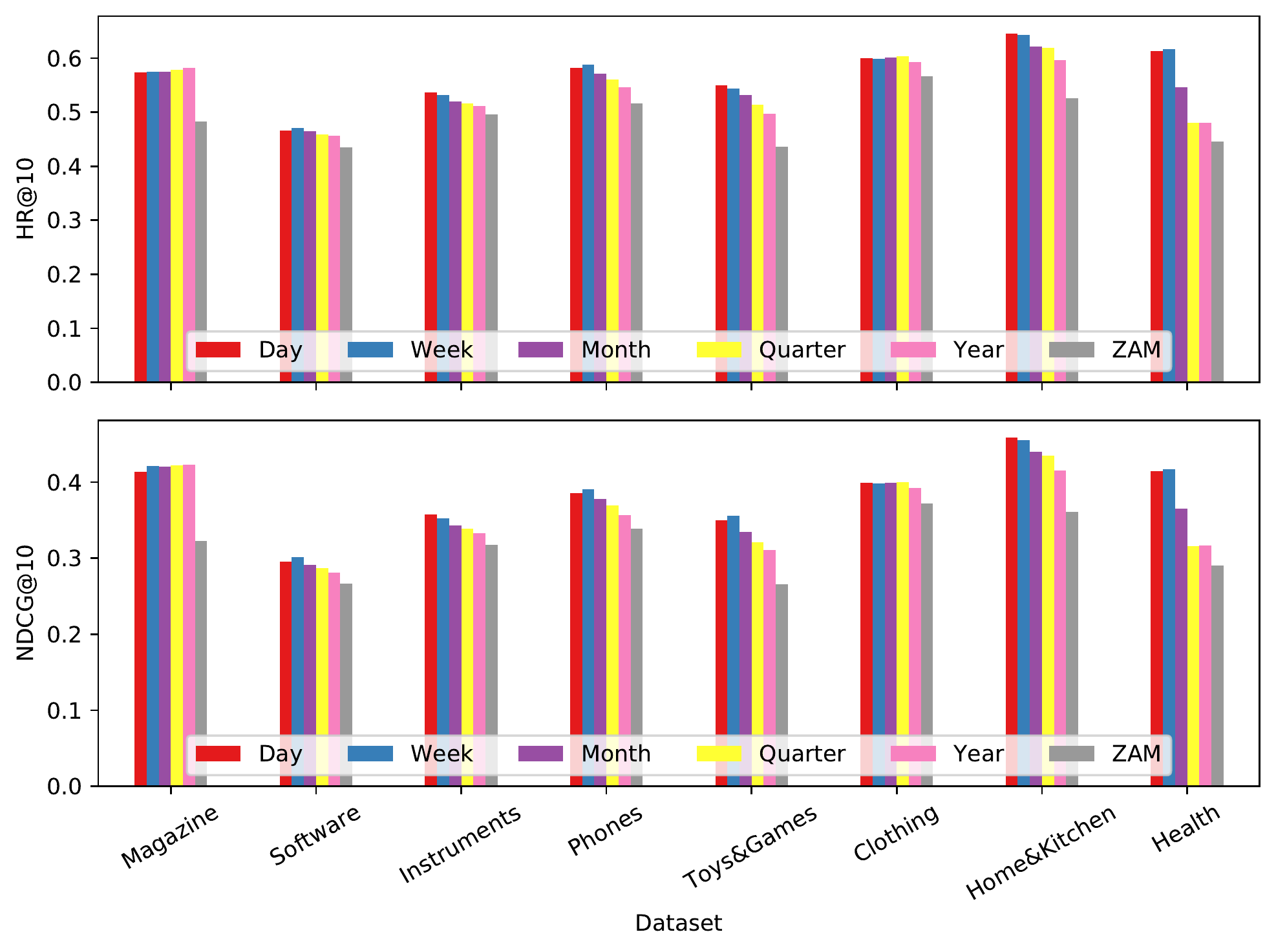}
  \caption{Ablation study of the time interval $R$ for successive graph construction. The results in HR@10 and NDCG@10 are reported. $R$ is varied from a day to a year. }
  \label{fig: time_scale}
\end{figure}

\begin{figure}[t]
  \centering
  \includegraphics[width=0.6\linewidth]{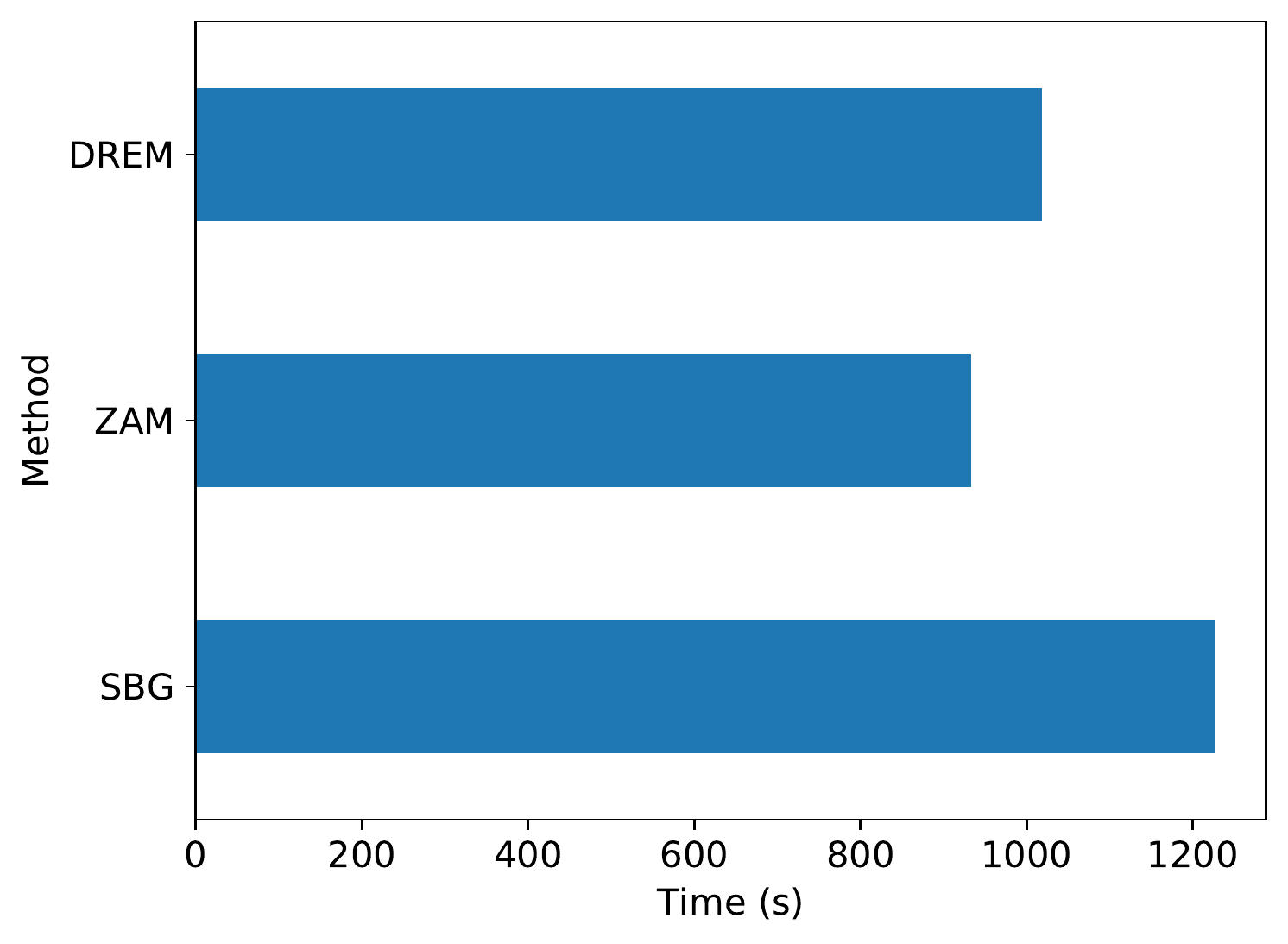}
  \caption{Comparison of the average training time (in seconds) on \textit{Home\&Kitchen}. 
  }
  \label{fig: runtime}
\end{figure}

\subsubsection{\textbf{Runtime Analysis}}

We evaluate the time efficiency of our SBG and the baselines on \emph{Home\&Kitchen}, the largest dataset used in the experiments. As shown in Figure~\ref{fig: runtime}, the running time of our method is comparable with that of ZAM and DREM. All the experiments are conducted on a platform with Intel(R) Xeon(R) Gold 6226R CPU @ 2.90GHz and GeForce RTX 3090.

\section{Conclusion} \label{sec:conclusion}

In this work, we have proposed a generic approach to model user preferences for personalized product search. To exploit local and global user behavior patterns for search personalization, 
we construct a successive behavior graph and capture implicit user preference signals with an efficient jumping graph convolution. 
Our approach can be used as a plug-and-play module in the popular latent space based product search framework and potentially in many other methods to improve their performance. 
Extensive experiments on public datasets demonstrate the effectiveness of our approach. In future work, we plan to investigate the possibility of applying our approach on dynamic behavior graphs for user preference modeling.

\section*{Acknowledgement}
We would like to thank the anonymous reviewers for their helpful comments. This research was supported by the grants of  P0034058 (ZGAL) and P0038850 (ZGD1) funded by Alibaba and the General Research Fund No.15222220 funded by the UGC of Hong Kong.

\bibliographystyle{ACM-Reference-Format}
\bibliography{main}

\clearpage
\newpage
\appendix
\section*{Appendix}
\setcounter{thm}{0}
\newcommand{\A}{\bm A}
\newcommand{\I}{\bm I}
\newcommand{\D}{\bm D}
\newcommand{\La}{\bm L}
\newcommand{\Tr}{\text{Tr}}
\newcommand{\F}{\bm F}
\newcommand{\Hi}{\bm H}

In this section, we provide a proof  for Theorem \ref{thm:jump}. For a better reading experience, we rewrite Eq. (\ref{eq:simple},\ref{eq:laplacian},\ref{eq:jump_enriched_product}) here. \\
The Laplacian-Beltrami operator, which measures the diversity of embeddings, is defined as:
\begin{equation}
    \label{eq:ap_laplacian}
    \Omega(\Hi) = \sum_{k,i,j} a_{ij} (\Hi_{ki} - \Hi_{kj})^2.
\end{equation}
The convolution without jumping connection is defined as:
\begin{equation}\label{eq:ap_simple}
    \Hi^{(l+1)}=\left(\omega\bm I + (1-\omega)\bm D^{-1}\bm A\right) \Hi^{(l)}.
\end{equation}
The convolution with jumping connection is defined as:
\begin{equation}
    \tilde{\Hi}^{(l+1)}=\left(\omega\bm I + (1-\omega)\bm D^{-1}\bm A\right) \left (\beta \Hi^{(0)} + (1-\beta)\tilde{\Hi}^{(l) } \right)\label{eq:ap_jump_enriched_product},
\end{equation}

\begin{thm}
If the initial diversity $\Omega(\Hi^{(0)}) >0$, then for any integer $l > 0$, $\beta \in (0,1)$, and $\omega \in (0.5, 1)$, $\tilde{\Hi}^{(l)}$ is strictly more diverse than $\Hi^{(l)}$:
\begin{equation}\label{eq:k-th}
  \Omega\left(\tilde{\bm H}^{(l)}\right) > \Omega\left({\bm H}^{(l)}\right).
\end{equation}
While $l$ approaches infinity, jumping connections can prevent the diversity of product representations from collapsing to 0 (over-smoothing), i.e.,
\begin{equation}\label{eq:limits}
    \lim_{l\to0} \Omega\left(\tilde{\Hi}^{(l)}\right) > \lim_{l\to0}\Omega\left(\Hi^{(l)}\right) = 0.
\end{equation}
\end{thm}

\begin{proof}

Let $\La = \I - \D^{-1}\A$, then $\Omega(\cdot)$ becomes
\begin{equation}
  \Omega(\Hi) = \sum_k \Hi_{:,k}^\top\La\Hi_{:,k},  
\end{equation}
where $\Hi_{:,k}$ is the $k$-th column of $\Hi$. Denote arbitrary column of $\Hi$ by $h$, then we only need to prove 
$$(\tilde{h}^{(l)})^\top \La \tilde{h}^{(l)} = \Omega(\tilde{h}^{(l)}) > \Omega(h^{(l)}) = (h^{(l)})^\top \La h^{(l)}$$
and
$$\lim_{l\to0} \Omega(\tilde{h}^{(l)}) > \lim_{l\to0} \Omega(h^{(l)}) = 0.$$

Let $\F = \left(\omega\bm I + (1-\omega)\bm D^{-1}\bm A\right)$. 
From Eq. (\ref{eq:ap_simple}) and (\ref{eq:ap_jump_enriched_product}), we could obtain general formula for $h^{(l)}$ and $\tilde{h}^{(l)}$:
\begin{align}
    h^{(l)} =& \F^l h^{(0)} \label{eq:ap_no_jump_f}\\
    \tilde{h}^{(l)} =& \left((1-\beta)^{l}\F^l + \beta\sum_{k=1}^{l} (1-\beta)^{k-1}\F^k\right)h^{(0)} \label{eq:ap_jump_f}
\end{align}

Denote the eigen-decomposition of $\La$ by $\bm U \bm \Lambda \bm U^\top$, where $\bm U$ is the eigenbasis and $\bm \Lambda$ is a diagonal matrix with corresponding eigenvalues, then
\begin{equation}\label{eq:ap_F}
    \F = (\I-(1-\omega) \La) = \bm U (\I-(1-\omega)\bm \Lambda)\bm U^\top
       =  \bm U \bm M \bm U^\top
\end{equation}
where $\bm M = \I-(1-\omega)\bm \Lambda$. Denote $c=\bm U^\top h^{(0)}$ and substitute $\F$ in Eq. (\ref{eq:ap_no_jump_f},\ref{eq:ap_jump_f}) by (\ref{eq:ap_F}):
\begin{align}
    \Omega(h^{(0)}) =& c^\top\bm\Lambda c = \sum_i \lambda_i c_i^2, \\
    \Omega(h^{(l)}) =& c^\top \bm M^{2l}\bm \Lambda c = \sum_i (1-(1-\omega)\lambda_i)^{2l}\lambda_ic_i^2 \\
    =& \sum_i g_l^2(\lambda_i)\lambda_i c_i^2, \label{eq:omega_no_jump}\\
    \Omega(\tilde{h}^{(l)}) 
    =& c^\top \left((1-\beta)^{l}\bm M^l + \beta\sum_{k=1}^{l} (1-\beta)^{k-1}\bm M^k\right)^2 \bm\Lambda c \\
    =& \sum_i f_l^2(\lambda_i)\lambda_i c_i^2, \label{eq:omega_jump}
\end{align}
where 
\begin{align}
    g_l(\lambda) =& (1-(1-\omega)\lambda)^l, \\ 
    f_l(\lambda) =& (1-\beta)^lg_l(\lambda) + \beta\sum_{k=1}^l(1-\beta)^{k-1}g_k(\lambda).
\end{align}
Now, we only need to compare $g_l(\lambda)$ and $f_l(\lambda)$.  Notice that $\omega\in(0.5, 1)$ and $\lambda$ is the eigenvalue of the normalized graph laplacian $\La$, so $\lambda \in [0,2]$ and $1-(1-\omega)\lambda \in (0,1]$.
Then, $g_l(\lambda)$ is positive and decreases as $l$ increases:
\begin{equation}
    g_{l_1}(\lambda) \ge  g_{l_2}(\lambda) > 0, \;\text{for any}\; l_2 > l_1.
\end{equation}
The equality holds only if $\lambda=0$. When it comes to $f_l(\lambda)$, we have
\begin{equation}
    f_l(\lambda) \ge \left((1-\beta)^l+\beta\sum_{k=1}^l(1-\beta)^{k-1}\right)g_l(\lambda) = g_l(\lambda) > 0. \label{eq:fg}
\end{equation}
Given Eq. (\ref{eq:omega_no_jump},\ref{eq:omega_jump},\ref{eq:fg}), we can conclude 
\begin{equation}
    \Omega(\tilde{h}^{(l)}) \ge \Omega(h^{(l)}).
\end{equation}
Notice that the initial diversity $\Omega(h^{(0)}) > 0$, so there exists such $\lambda_i$ that $\lambda_i c_i^2 > 0$ and $\lambda_i \not= 0$, and the equality does not hold and the inequality (\ref{eq:k-th}) is proved. 

Now, we consider the limits of $f_l(\lambda)$ and $g_l(\lambda)$:
\begin{align}
    \lim_{l\to\infty} f_l(\lambda)= & \frac{\beta(1-(1-\omega)\lambda)}{1-(1-\beta)(1-(1-\omega)\lambda)} > 0, \\
    \lim_{l\to\infty} g_l(\lambda)= & 0, \quad \text{for all}\; \lambda > 0.
\end{align}
As a consequence, 
\begin{equation}
    \lim_{l\to0} \Omega(\tilde{h}^{(l)}) > \lim_{l\to0} \Omega(h^{(l)}) = 0.
\end{equation}
Eq. (\ref{eq:limits}) is also proved.
\end{proof}

\end{document}